\documentclass{easychair}

\usepackage{multirow}
\usepackage{color}

\newcommand{\OMIT}[1]{}

\newcommand{\PresQF}{\text{QF}(\mathbb{N})}
\newcommand{\PresEx}{\exists^*(\mathbb{N})}
\newcommand{\Logic}{\mathcal{L}}

\newcommand{\defn}[1]{\emph{#1}}
\newcommand{\N}{\mathbb{N}}

\usepackage{stmaryrd}
\usepackage{amssymb}

\newcommand\set[1]{\left\{#1\right\}}
\newcommand\ap[2]{{#1}\mathord{\brac{#2}}}
\newcommand\brac[1]{\left({#1}\right)}
\newcommand\tup[1]{\left({#1}\right)}
\newcommand\setcomp[2]{\left\{{#1}\ \left|\ {#2}\right.\right\}}
\newcommand\minsat[1]{\constc_{#1}}
\newcommand\semilin[2]{\ap{L}{#1; #2}}
\newcommand\maxiof[1]{\ap{\mathrm{max}}{#1}}
\newcommand\solsof[1]{\left\llbracket #1 \right\rrbracket}

\newcommand\isreg[1]{\mathrm{Region}_{#1}}
\newcommand\auxfplus[1]{\varz^{+}_{#1}}
\newcommand\auxfminus[1]{\varz^{-}_{#1}}

\newcommand\reggeq[1]{\geq_{#1}}
\newcommand\regequiv[2]{=^{#2}_{#1}}
\newcommand\eqreps[2]{C^{#2}_{#1}}
\newcommand\abs[1]{\ap{\mathrm{abs}}{#1}}
\newcommand\divsat[1]{\mathrm{Div}_{#1}}

\newcommand\pres{\phi}
\newcommand\varpres{\psi}
\newcommand\monpres{\Delta}
\newcommand\consta{a}
\newcommand\constb{b}
\newcommand\constc{c}
\newcommand\constd{d}
\newcommand\denom{k}
\newcommand\idxi{i}
\newcommand\idxj{j}
\newcommand\varx{x}
\newcommand\vary{y}
\newcommand\varz{z}
\newcommand\eqdiv[1]{\equiv_{#1}}
\newcommand\nats{\mathbb{N}}
\newcommand\ints{\mathbb{Z}}
\newcommand\numof{n}
\newcommand\varnumof{m}
\newcommand\vecx{\overline{\varx}}
\newcommand\vecz{\overline{\varz}}
\newcommand\vecc{\overline{\constc}}
\newcommand\vecd{\overline{\constd}}
\newcommand\propf{S}
\newcommand\veca{\overline{\consta}}
\newcommand\vecb{\overline{\constb}}
\newcommand\vecp{\overline{p}}
\newcommand\vecy{\overline{\vary}}
\newcommand\bound{B}
\newcommand\samediv{\mathrm{SameDiv}}
\newcommand\alldivshort{\mathrm{Divs}}
\newcommand\alldivlong{\mathrm{XDivs}}
\newcommand\divset{D}
\newcommand\badx{\chi}
\newcommand\linearf{f}
\newcommand\linearg{g}

\newcommand\xfs{F}
\newcommand\regmap{{\hat{r}}}
\newcommand\polyf{g}
\newcommand\baseset{A}
\newcommand\periodset{P}
\newcommand\auxeq{\mathrm{Aux}}
\newcommand\vecaux{\overline{\varz}_\xfs}
\newcommand\veceaux{\overline{e}_\xfs}

\newcommand\regub{\mathrm{UB}}
\newcommand\regeq{\mathrm{EQ}}
\newcommand\eqdef{\triangleq}

\newtheorem{definition}{Definition}[section]
\newtheorem{theorem}[definition]{Theorem}
\newtheorem{lemma}[definition]{Lemma}
\newtheorem{proposition}[definition]{Proposition}

\newcommand\shortlong[2]{#2}

\title{
    Monadic Decomposition in Integer Linear Arithmetic \\
    (Technical Report)
}
\author{
    Matthew Hague\inst{1}\thanks{Orcid ID: 0000-0003-4913-3800} \and
    Anthony W. Lin\inst{2}\thanks{Orcid ID: 0000-0003-4715-5096} \and
    Philipp R\"ummer\inst{3}\thanks{Orcid ID: 0000-0002-2733-7098} \and
    Zhilin Wu\inst{4}
}
\institute{
    Royal Holloway, University of London, United Kingdom \and
    TU Kaiserslautern, Germany \and
    Uppsala University, Sweden \and
    State Key Laboratory of Computer Science, \\
    Institute of Software, Chinese Academy of Sciences, China
}

\authorrunning{M. Hague, A. W. Lin, P. R\"ummer, Z. Wu}
\titlerunning{Monadic Decomposition in Integer Linear Arithmetic}

\begin{document}

\maketitle

\begin{abstract}
    Monadic decomposability is a notion of variable independence, which asks
    whether a given formula in a first-order theory is expressible as a Boolean
    combination of monadic predicates in the theory. Recently, Veanes et al.
    showed the usefulness of monadic decomposability in the context of SMT
    (i.e. the input formula is quantifier-free),
    and found various interesting applications including string analysis.
    However, checking monadic decomposability is undecidable in
    general. Decidability for certain theories is known (e.g. Presburger
    Arithmetic, Tarski's Real-Closed Field), but there are very few results
    regarding their computational complexity.
    In this paper, we study monadic decomposability of integer linear arithmetic
    in the setting of SMT.
    We show that this decision problem is coNP-complete and,
    when monadically decomposable, a formula admits a decomposition of
    exponential size in the worst case.
    We provide a new application of our results to
    string constraint solving with length constraints.
    We then extend our results to variadic decomposability,
    where predicates could admit multiple free variables (in contrast to monadic
    decomposability). Finally, we give an
    application to quantifier elimination in integer linear arithmetic where
    the variables in a block of quantifiers, if independent,
    could be eliminated with an exponential (instead of the standard
    doubly exponential) blow-up.

\end{abstract}


\section{Introduction}
\label{sec:intro}

A formula $\phi(\bar x)$ in some theory $\Logic$ is
\defn{monadically decomposable} if it is $\Logic$-equivalent to a
Boolean combination of monadic predicates in $\Logic$, i.e., to a
\defn{monadic decomposition} of $\phi$. Monadic decomposability
measures how tightly the free variables in $\phi$ are coupled. For
example, $x = y$ is not monadically decomposable in any (finitary)
logic over an infinite domain, but $x + y \geq 2$ can be decomposed,
in Presburger arithmetic over natural numbers, since it can be written
as $x \geq 2 \vee (x \geq 1 \wedge y \geq 1) \vee y \geq 2$.

Veanes \emph{et al.}  \cite{monadic-decomposition} initiated the study
of monadic decomposability in the setting of Satisfiability Modulo
Theories, wherein formulas are required to be quantifier-free. Monadic
decomposability has many applications, including symbolic transducers
\cite{symbolic-transducer-power} and string analysis
\cite{monadic-decomposition}. Although the problem was shown to be in
general undecidable, a generic semi-algorithm for outputting monadic
decompositions (if decomposable) was provided.  A termination check
could in fact be added if the input formula belongs to a theory for
which monadic decomposability is decidable, e.g., linear arithmetic,
Tarski's Real-Closed Field, and the theory of uninterpreted
functions.
Hitherto, not much is known about the computational complexity of
monadic decomposability problems for many first-order theories (in particular,
quantifier-free theories), and about practical algorithms.
This was an open problem raised by Veanes \emph{et al.} in \cite{monadic-decomposition}.

Monadic decomposability is intimately connected to the variable partition
problem, first studied by Libkin \cite{Libkin03} nearly 20 years ago. In
particular, a monadic
decomposition gives rise to a partition of the free variables $\bar x$ of a
formula $\phi(\bar x)$,
wherein each part consists of a single variable. More precisely,
take a partition $\Pi = \{Y_1,\ldots,Y_m\}$ of $\bar x$
into sets $Y_i$ of variables, with linearizations $\vecy_i$.
The formula $\phi(\bar x)$ is $\Pi$-decomposable (in some
theory $\Logic$)
if it is $\Logic$-equivalent to a boolean combination of formulas of the
form
$\Delta(\vecy_i)$.
As suggested in \cite{Libkin03}, such \emph{variadic decompositions} of
$\phi(\bar x)$ have potential applications in optimization of database query
processing and quantifier elimination. The author gave a general condition for the decidability
of variable independence in first-order theories. This result is unfortunately
not easily applicable in the SMT setting for at least two reasons: (i) the full
first-order theory might be undecidable (e.g. theory of uninterpreted
functions), and (ii) even for a first-order theory that admits decidable monadic
decompositions, the complexity of the algorithm obtained from \cite{Libkin03}
could be too prohibitive for the quantifier-free fragment. One example that
epitomizes (ii) is the problem of determining whether a given relation $R
\subseteq (\Sigma^*)^k$ over strings represented by a regular transducer
could be expressed as a boolean combination of monadic predicates. The result of
\cite{Libkin03} would give a double exponential-time algorithm for monadic
decomposability, whereas it was recently shown in \cite{BHLLN19}
to be solvable in polynomial-time (resp. polynomial-space) when the transducer
is given as a deterministic (resp. nondeterministic) machine.

\paragraph{Contributions.}
First, we determine the complexity of
deciding monadic
decomposability and outputting monadic decompositions (if they exist) for the theory
of integer linear arithmetic in the setting of SMT. Our result is summarized
in Theorem \ref{th:main}.


\begin{theorem}[Monadic Decomposability] \label{thm:mondec}
    Given a quantifer-free formula $\pres$ of Presburger
Arithmetic, it is coNP-complete to decide if $\pres$ is
monadically decomposable.
    This is efficiently reducible
to unsatisfiability of quantifier-free Presburger formulas.
    Moreover, if a decomposition exists, it can be constructed in
    exponential time.  
    \label{th:main}
\end{theorem}
We show a new application of monadic decomposability in integer linear
arithmetic for SMT over strings, which is currently a
very active research area, e.g., see
\cite{trau18,CHLRW18,Berkeley-JavaScript,cvc4,S3,Abdulla14,Z3-str3,AGS18,philipp-survey,LB16,florin-sat}. One problem that makes string constraint solving difficult is the
presence of additional \emph{length constraints,} which forces the lengths of the
strings in the solutions to satisfy certain linear arithmetic constraints.
Whereas satisfiability of string equations with regular constraints is
PSPACE-complete (e.g. see \cite{J17,diekert}), it is a long-standing open
problem \cite{Vijay-length,buchi} whether word equations with length
constraints are decidable.
Length constraints are omnipresent in Kaluza~\cite{Berkeley-JavaScript},
arguably the first serious string constraint benchmarks obtained from real-world
JavaScript applications. Using our monadic decomposability solver, we show that 90\% of the Kaluza benchmarks are in fact
in a decidable fragment of string constraints, since occurring
length constraints can be completely removed by means of decomposition.

Next we extend our result to variadic decomposability (cf. \cite{Libkin03}).
\begin{theorem}[Variadic Decomposability] \label{thm:variadic}
    It is coNP-complete to decide if $\pres$ is $\Pi$-decomposable,
    given a quantifer-free formula $\pres(\bar x)$ of
    Presburger Arithmetic and a partition $\Pi = \{Y_1,\ldots,Y_n\}$ of
    $\bar x$.
    This is efficiently reducible to unsatisfiability of quantifier-free
    Presburger formulas.
    Moreover, if a decomposition exists, it can be constructed in
    exponential time. 
    \label{th:main2}
\end{theorem}
We show how this could be applied to quantifier elimination. In
particular, we show that if a formula
$\phi(\bar y) = \exists \bar x . \, \psi(\bar x,\bar y)$, where
$\psi$ is quantifier-free, is $\{X, Y\}$-decomposable---where $\bar x$
and $\bar y$ are linearizations of the variables in $X$ and $Y$---then
we can compute in exponential time a formula $\theta(\bar y)$ such
that
$\langle \mathbb{N},+\rangle \models \theta \leftrightarrow \phi$,
i.e., avoiding the standard double-exponential blow-up (cf.\
\cite{Weis97}).

\paragraph{Organization.} Preliminaries are in Section \ref{sec:prelim}.
Results on monadic (resp. variadic) decomposition are in
Section \ref{sec:mondec} (resp. Section \ref{sec:vardec}) and applications appear in Section \ref{sec:app}.



\section{Preliminaries}
\label{sec:prelim}

\subsection{Presburger Syntax}

In this paper we study the problem of monadic decomposition for
formulas in linear integer arithmetic. All of our results are
presented for Presburger arithmetic over natural numbers, but they
can be adapted easily to all integers.

\begin{definition}[Fragments of Presburger Arithmetic]
    A formula $\pres$ of \emph{Presburger arithmetic} is a formula of the form
    $\mathcal{Q}_1x_1\cdots \mathcal{Q}_nx_n . \, \psi$  where $\mathcal{Q}_i
    \in\{\forall,\exists\}$ and $\psi$ is a quantifier-free Presburger formula:
    \[
        \psi := \sum_\idxi \consta_\idxi \varx_i \sim \constb
              \ |\ \consta \varx \eqdiv{\denom} \constb \vary
              \ |\ \varx \eqdiv{\denom} \constc
              \ |\ \pres_1 \land \pres_2
              \ |\ \neg \pres
    \]
    where
        $\consta_\idxi, \consta, \constb \in \ints$,
        $\denom, \constc \in \nats$ with $0 \leq \constc < \denom$,
        variables $\varx_\idxi, \varx, \vary$ range over $\nats$, and
        ${\sim} \in \set{\leq, \geq}$.
    The operator $\eqdiv{\denom}$ denotes equality modulo $\denom$,
    i.e., $s \eqdiv{\denom} t$ whenever $s - t$ is a multiple of $\denom$.
    Formulas of the shape $\sum_\idxi \consta_\idxi \varx_i \sim \constb$,
    $\consta \varx \eqdiv{\denom} \constb \vary$, or
    $\varx \eqdiv{\denom} \constc$
    are called \emph{atoms.}

    Existential Presburger formulas are formulas of
    the form $\exists x_1,\ldots,x_n . \, \psi$
    for some quantifier-free Presburger formula~$\psi$.
    We let $\PresQF$ (resp. $\PresEx$) denote the set of all
    quantifier-free (resp., existential) Presburger formulas.
\end{definition}

Let
$\vecx = \tup{\varx_1, \ldots, \varx_\numof}$
be a tuple of integer variables.
We write
$\ap{\linearf}{\vecx} = \sum_\idxi \consta_\idxi \varx_\idxi$
for a linear sum over $\vecx$.
Let
$\vecy = \tup{\vary_1, \ldots, \vary_\varnumof}$.
By slight abuse of notation, we may also write
$\ap{\pres}{\vecx, \vecy}$
to denote a $\PresQF$ formula over the variables
$\varx_1, \ldots, \varx_\numof, \vary_1, \ldots, \vary_\varnumof$.

\subsection{Monadic Decomposability}
\label{sec:mon-decomp-prelims}

A quantifier-free formula~$\pres$ is called \emph{monadic} if every
atom in $\pres$ contains at most one variable, and it is called
\emph{monadically decomposable} if $\pres$ is equivalent to a monadic
formula~$\pres'$. In this case, $\pres'$ is also called a
\emph{decomposition} of $\pres$. For our main results we use a
slightly refined notion of a formula being decomposable:
\begin{definition}[Monadically Decomposable on $\varx$]
    Fix a logic $\Logic$ (e.g. $\PresQF$ or $\PresEx$). We say a formula
    $\ap{\pres}{\varx_1, \ldots, \varx_\numof}$ in $\Logic$
    is \emph{monadically decomposable on $\varx_\idxi$} whenever
    \[
        \ap{\pres}{\varx_1, \ldots, \varx_\numof}
        \equiv
        \bigvee\limits_{\idxj}
            \ap{\monpres_\idxj}{\varx_\idxi}
            \land
            \ap{\varpres_\idxj}{\varx_1, \ldots, \varx_{\idxi-1},
                                \varx_{\idxi+1}, \ldots, \varx_\numof}
    \]
    for some formulas $\monpres_\idxj$ and $\varpres_\idxj$ in $\Logic$.
\end{definition}

It can be observed that a formula is monadically decomposable if and only
if it is monadically decomposable on all variables occurring in the
formula (cf. Lemma \ref{lm:single_vs_all}). We expand on this for the variadic
case below.

We recall the following characterization of monadic decomposability for formulas
$\phi(x,y)$ with two free variables (cf.
\cite{CCG06,monadic-decomposition,BHLLN19,Libkin03}), which holds regardless of the
theory under consideration. This can be extended easily
to formulas with $k$ variables, but is not needed in this paper. Given a formula
$\phi(x,y)$, define the formula $\sim$ as follows:
\[
    x \sim x' := \forall y, y'.\ ( \phi(x,y) \wedge \phi(x',y') \to
    (\phi(x',y) \wedge \phi(x,y')))
\]
\begin{proposition}
    The relation $\sim$ is an equivalence relation. Furthermore,
    $\phi(x,y)$ is monadically decomposable iff $\sim$ has a finite index
    (i.e. the number of $\sim$-equivalence classes is finite).
    \label{prop:char}
\end{proposition}
Using this proposition, it is easy to show that over a structure with
an infinite domain (e.g. integer linear arithmetic) the formula $x=y$
is not monadically decomposable. As was noted already in
\cite{Libkin03}, to check monadic decomposability of a
formula~$\phi$ in Presburger Arithmetic in general, we may simply
check if there is an upper bound $B$ on the smallest representation of
every $\sim$-equivalence class, i.e.,
\begin{equation*}
  \exists B. \forall x. \exists x_s.\ ( x_s \leq B \wedge x_s \sim x )~.
\end{equation*}
However, to derive tight complexity bounds for checking monadic
decomposability, this approach is problematic, since the above
characterisation has multiple quantifier alternations. Using known results
(e.g. \cite{Haase14}), one would only obtain
an upper bound in the weak exponential hierarchy \cite{Haase14}, which
only admits double-exponential time algorithms.

\subsection{Variadic Decomposability}
\label{sec:var-decomp-prelims}

The notion of a \emph{variadic decomposition} generalises monadic
decomposition by considering partitions of the occurring
variables.
\begin{definition}[$\Pi$-Decomposable]
    Fix a logic $\Logic$ (e.g. $\PresQF$ or $\PresEx$). Take a formula
    $\ap{\pres}{\varx_1, \ldots, \varx_\numof}$ in $\Logic$ and a partition
    $\Pi = \{Y_1,\ldots,Y_\varnumof\}$
    of
    $\varx_1, \ldots, \varx_\numof$.
    We say $\pres$ is \emph{$\Pi$-decomposable} whenever
    \[
        \ap{\pres}{\varx_1, \ldots, \varx_\numof}
        \equiv
        \bigvee\limits_{\idxi}
            \ap{\monpres^1_\idxi}{\vecy_1}
            \land
            \cdots
            \land
            \ap{\monpres^\varnumof_\idxi}{\vecy_\varnumof}
    \]
    for some formulas $\monpres^\idxj_\idxi$ in $\Logic$ and linearizations $\vecy_\idxj$ of $Y_\idxj$.
\end{definition}

Observe that a formula
$\ap{\pres}{\varx_1, \ldots, \varx_\numof}$
is monadically decomposable on $\varx_\idxi$ iff it is $\Pi$-decomposable with
$\Pi = \{
    \{\varx_\idxi\},
    \{\varx_1, \ldots, \varx_{\idxi-1},
      \varx_{\idxi+1}, \ldots, \varx_\numof\}
\}$.
Moreover, we say a formula $\pres$ over the set of variables $X$
is \emph{variadic decomposable on $Y$} whenever it is $\Pi$-decomposable with
$\Pi = \{Y, X \setminus Y\}$.

General $\Pi$-decompositions can be computed by
decomposing on binary partitions~$\{Y, X \setminus Y\}$, which is why
we focus on this binary case in the rest of the paper.
We argue why this is the case below.

Let a formula~$\phi$ and $\Pi = \{Y_1, \ldots, Y_\varnumof\}$ be given.
We can first decompose separately on each
$\{Y_\idxi, Y\}$
where
$Y = Y_1 \cup \cdots \cup Y_{\idxi-1} \cup Y_{\idxi+1} \cup \cdots \cup Y_{\varnumof}$.
Using the algorithm in Section~\ref{sec:vardec} we obtain for each $\idxi$ a decomposition of a specific form:
\[
    \bigvee\limits_{\idxj}
        \ap{\monpres^\idxi_\idxj}{\vecy_\idxi}
        \land
        \ap{\phi}{\vecy_1, \ldots, \vecy_{\idxi-1},
                  \vecc^\idxi_\idxj,
                  \vecy_{\idxi+1}, \ldots, \vecy_{\varnumof}}
    \ .
\]
Note, these decompositions can be performed independently using the algorithm in Section~\ref{sec:vardec} and the second conjunct of each disjunct is $\phi$ with $\vary_\idxi$ replaced by fixed constants $\vecc_\idxj$.
Additionally, each
$\monpres^\idxi_\idxj$
is polynomial in size and each
$\vecc^\idxi_\idxj$
can be represented with polynomially many bits.
We note also that our algorithm ensures that each
$\monpres^\idxi_\idxj$
is satisfiable.

Given such decompositions, we can recursively decompose $\phi$ on $\Pi$.
We first use the above decomposition for $\idxi = 1$ and obtain
\[
    \bigvee\limits_{\idxj}
        \ap{\monpres^1_\idxj}{\vecy_1}
        \land
        \ap{\phi}{\vecc^1_\idxj, \vecy_2, \ldots, \vecy_{\varnumof}}
    \ .
\]

Next, we use the decomposition for $\idxi = 2$ to decompose the copies of $\phi$ in the decomposition above.
We obtain
\[
    \bigvee\limits_{\idxj_1} \brac{
        \ap{\monpres^1_{\idxj_1}}{\vecy_1}
        \land
            \bigvee\limits_{\idxj_2}
                \ap{\monpres^2_{\idxj_2}}{\vecy_2}
                \land
                \ap{\phi}{\vecc^1_{\idxj_1},
                          \vecc^2_{\idxj_2},
                          \vecy_3, \ldots, \vecy_{\varnumof}}
    } \ .
\]

This process repeats until all $Y_\idxi$ have been considered.
If $\pres$ is $\Pi$-decomposable, we find a decomposition.
If $\pres$ is not $\Pi$-decomposable, then it would not be possible to do the independent decompositions for each $\idxi$.
Thus, for $\Pi = \{Y_1,\ldots, Y_m\}$, we can use variadic decompositions on $Y_i$ to compute $\Pi$-decompositions.

The above algorithm runs in exponential time due both to the exponential size of the decompositions and the branching caused by the disjuncts.
If we are only interested in whether a formula is $\Pi$-decomposable, it is
enough to ask whether it is decomposable on $Y_\idxi$ for \emph{each} $\idxi$.
In particular, a formula $\phi(\bar x)$ is monadically decomposable iff
$\phi$ is decomposable for each variable $y \in \bar x$. Since the complexity
class coNP is closed under intersection, we obtain the following:
\begin{lemma}
    A coNP upper bound for monadic decomposability on a given
    variable $y$ implies a coNP upper bound for monadic decomposability.
    Likewise, a coNP upper bound for variadic decomposability on a given subset
    $Y$ of variables implies a coNP upper bound for $\Pi$-decomposability.
    \label{lm:single_vs_all}
\end{lemma}

\subsection{Example}

Consider the formula
$\ap{\pres}{\varx, \vary, \varz}$
given by
$\varz = \varx + 2 \vary \land \varz < 5$.
This formula is monadically decomposable, which means, it is $\Pi$-decomposable for
$\Pi = \{\{x\}, \{y\}, \{z\}\}$.

Our algorithm will first take a decomposition on $\varx$ and might obtain
$\bigvee^4_{\idxi = 0}
    \ap{\monpres^1_\idxi}{\varx}
    \land
    \ap{\pres}{\idxi, \vary, \varz}$
where
$\ap{\monpres^1_\idxi}{\varx} = (\varx = \idxi)$
and
$\ap{\pres}{\idxi, \vary, \varz} =
    (\varz = \idxi + 2y)
    \land
    \varz < 5$.
Next, we use a decomposition on $\vary$.
For each $\ap{\pres}{\idxi, \vary, \varz}$ we substitute
$
    \bigvee ^{
        2 - \lceil \frac{\idxi}{2} \rceil
    }_{\idxj = 0}
        \vary = \idxj
        \land
        \ap{\pres}{\idxi, \idxj, \varz}
$, and as the  final decomposition we get
\begin{equation*}
\bigvee\limits^4_{\idxi = 0}
    \bigvee\limits^{
        2 - \lceil \frac{\idxi}{2} \rceil
    }_{\idxj = 0}
    \varx = \idxi
        \land
        \vary = \idxj
        \land
        \varz = \idxi + 2 \idxj ~.
\end{equation*}



\section{Monadic Decomposability}
\label{sec:mondec}

\subsection{Lower Bounds}

We first show that unsatisfiability of Boolean formulas can be reduced
to monadic decomposability of formulas with only two variables,
directly implying coNP-hardness:
\begin{lemma}[coNP-Hardness]
    Deciding whether a formula~$\pres(x,y)$ in $\PresQF$ is monadic
    decomposable is coNP-hard.
\end{lemma}
\begin{proof}
    We reduce from unsatisfiability of propositional formulas to monadic
    decomposability of $\pres(x,y)$. Take a propositional formula
    $\ap{\propf}{\varx_1, \ldots, \varx_\numof}$. Let $p_1,\ldots,p_\numof$ be
    the first $\numof$ primes. Let $\psi(x)$ be the formula obtained from
    $\propf$ by replacing each occurrence of $\varx_i$ by $x \eqdiv{p_i} 0$.
    Given an assignment $\nu: \{x_1,\ldots,x_n\} \to \{0,1\}$, we let
    \[
        H_\nu = \{ m \in \N \mid
        \forall 1 \leq i \leq n.\ ( m \eqdiv{p_i} 0 \leftrightarrow \nu(x_i) = 1) \}~.
            \]
            Thanks to the Chinese Remainder Theorem, $H_\nu$ is
            non-empty and periodic with
            period~$p = \prod_{i=1}^\numof p_i$, which implies that
            $H_\nu$ is infinite for every $\nu$. We also have that
            $\nu \models \propf$ iff, for each $n \in H_\nu$,
            $\psi(n)$ is true.

    Now define
    $\pres(x,y) = (\psi(x) \wedge x=y)$.
    If $\propf$ is unsatisfiable, then $\psi$ is unsatisfiable and so it is
    decomposable. Conversely, if $\propf$ can be satisfied by some
    assignment $\nu$, then $\pres(m,m)$ is true for all (infinitely many)
    $m \in H_\nu$. Since
    all solutions to $\pres(x,y)$ imply that $x=y$, by Proposition
    \ref{prop:char} we have that $\pres$ is
    not monadically decomposable.
    \qed
\end{proof}

We next provide exponential lower bounds for decompositions in either
\emph{disjunctive normal form} (DNF) or \emph{conjunctive normal form} (CNF).
DNF has been frequently used to represent monadic decompositions by
previous papers (e.g. \cite{Libkin03,BHLLN19,CCG06}), and it is most suitable
for applications in quantifier elimination.
\begin{lemma}[Size of Decomposition]
  There exists a family $\{\pres_n(\varx, \vary)\}_{n \in \nats}$ of formulas in
  $\PresQF$ such that $\pres_n$ grows
  linearly in $n$, while the smallest decomposition on $\varx$ in
  DNF/CNF is exponential in $n$.
\end{lemma}
\begin{proof}
  Consider the formulas
  $\ap{\pres_n}{\varx, \vary} = (\varx + \vary \leq 2^n)$.
  Using a binary encoding of constants, the size of the formulas is
  linear in $n$. We show that decompositions in DNF/CNF must be
  exponential in size.

  \emph{Disjunctive:} Suppose
  $ \ap{\psi_n}{\varx, \vary} = \bigvee_i \psi_i^x(x) \wedge
  \psi_i^y(y) $ is a monadic decomposition in DNF. Each disjunct
  $\psi_i^x(x) \wedge \psi_i^y(y)$, if it is satisfiable at all, has
  an upper right corner~$(x_i, y_i)$ such that
  $\psi_i^x(x_i) \wedge \psi_i^y(y_i)$ holds, but
  $\psi_i^x(x) \wedge \psi_i^y(y) \Rightarrow x \leq x_i \wedge y \leq
  y_i$. This immediately implies that exponentially many disjuncts are
  needed to cover the exponentially many points on the
  line~$x + y = 2^n$.

  \emph{Conjunctive:} Suppose $\psi_n(\varx, \vary)$ is a succinct
  monadic decomposition of $\pres_n$ in CNF. Since $\neg \psi_n(\varx, \vary) \equiv 2^n+1 \le x + y \equiv (2^n-x+1)+(2^n-y) \le 2^n$, it follows that $\neg \psi_n(2^n - x + 1, 2^n - y) \equiv (2^n-(2^n-x+1)+1)+(2^n-(2^n-y)) \le 2^n \equiv x+y \le 2^n$. Therefore, $\neg \psi_n(2^n - x + 1, 2^n - y)$ is a succinct decomposition of $\pres_n$ in
  DNF, contradicting the lower bound for DNFs.\qed
\end{proof}

\subsection{Upper Bound}
\label{sec:monadicUpper}

We prove Theorem~\ref{thm:mondec}. Following Lemma \ref{lm:single_vs_all},
it suffices to show that testing decomposability on a variable $\varx$ is in
coNP and that a decomposition can be computed in exponential time.
Assume without loss of generality that we have
$\ap{\pres}{\varx, \vecy}$
where
$\vecy = \tup{\vary_1, \ldots, \vary_\numof}$, and that
we are decomposing on the first variable~$x$.

We claim that $\pres$ is monadically decomposable on $\varx$ iff
\[
    \forall \varx_1, \varx_2 \geq \bound .
        \forall \vecy .~~
            \ap{\samediv}{\varx_1, \varx_2, \vecy}
            \Rightarrow
            \brac{
                \ap{\pres}{\varx_1, \vecy}
                \iff
                \ap{\pres}{\varx_2, \vecy}
            }
\]
where
    $\bound$ is a bound exponential in the size of $\pres$ and $\samediv$ is a formula asserting that $\varx_1$ and $\varx_2$ satisfy the same divisibility constraints.
This bound is computable in polynomial time and is described in Section~\ref{sec:bound-size-mon}.
To define $\samediv$, let $\alldivshort$ be the set of all divisibility constraints
$\consta \varz_1 \eqdiv{\denom} \constb \varz_2$
or
$\varz_1 \eqdiv{\denom} \constc$
appearing (syntactically) in $\pres$.
Assume without loss of generality that $\varx$ always appears on the left-hand side of a divisibility constraint (i.e., in the $\varz_1$ position of $\consta \varz_1 \eqdiv{\denom} \constb \varz_2$).
We then define
\[
    \ap{\samediv}{\varx_1, \varx_2, \vecy}
    ~=~~
    \brac{\begin{array}{c}
        \bigwedge\limits_{
            \consta x \eqdiv{\denom} \constb \varz
            \in
            \alldivshort
        }
            \brac{\consta \varx_1 \eqdiv{\denom} \constb \varz}
            \iff
            \brac{\consta \varx_2 \eqdiv{\denom} \constb \varz}
        \\ \land \\
        \bigwedge\limits_{
            x \eqdiv{\denom} \constc
            \in
            \alldivshort
        }
            \brac{\varx_1 \eqdiv{\denom} \constc}
            \iff
            \brac{\varx_2 \eqdiv{\denom} \constc}
    \end{array}} \ .
\]
We prove the claim in the following sections and simultaneously show how to construct the decomposition.
Once we have established the above, we can test non-decomposability on $\varx$ by checking
\[
    \exists \varx_1, \varx_2 \geq \bound . \,
        \exists \vecy . ~~
            \ap{\samediv}{\varx_1, \varx_2, \vecy}
            \land
            \ap{\pres}{\varx_1, \vecy}
            \land
            \neg \ap{\pres}{\varx_2, \vecy}
\]
which is decidable in NP.
Thus we obtain a coNP decision procedure because the above formula is polynomial in the size of $\pres$.

\subsubsection{Example}

We consider some examples.
First consider the formula
$\varx = \vary$
that cannot be decomposed on $\varx$.
Since there are no divisibility constraints, $\samediv$ is simply $\mathit{true}$.
It is straightforward to see that,
$\forall \bound. \exists \varx_1, \varx_2 \geq \bound .
    \exists \vary .~
        \mathit{true} \land \varx_1 = \vary \land \varx_2 \neq \vary$,
for example by setting
    $\varx_1 = \bound$,
    $\varx_2 = \bound + 1$, and
    $\vary = \bound$.

Now consider the monadically decomposable formula
\[
    \ap{\pres}{\varx, \vary, \varz} =
        \varx + 2 \vary \geq 5 \land \varz < 5 \land \varx \eqdiv{2} \vary
    \ .
\]
In this case
$\ap{\samediv}{\varx_1, \varx_2, \vary, \varz} =~
    (\varx_1 \eqdiv{2} \vary
     \iff
     \varx_2 \eqdiv{2} \vary)$.
We can verify
\[
    \forall \varx_1, \varx_2 \geq \bound .
        \forall \vary, \varz .~~
            \ap{\samediv}{\varx_1, \varx_2, \vary, \varz}
            \Rightarrow
            \brac{
                \ap{\pres}{\varx_1, \vary, \varz}
                \iff
                \ap{\pres}{\varx_2, \vary, \varz}
            }
\]
holds, as it will be the case that $5 < \bound$ and for all $\varx > 5$ the formula $\pres$ will hold whenever
$\varx \eqdiv{2} \vary$
holds and
$\varz < 5$.
The precondition $\samediv$ ensures that the if and only if holds.
We will construct the decomposition in the next section.

\subsubsection{Expanded Divisibility Constraints}

Observe that divisibility constraints are always decomposable.
In particular,
$\consta \varz_1 \eqdiv{\denom} \constb \varz_2$
is equivalent to a finite disjunction of clauses
$\varz_1 \eqdiv{\denom'} \constc \land
 \varz_2 \eqdiv{\denom'} \constc$
where $\denom'$ and $\constc$ are bounded by a multiple of $\consta, \constb$ and $\denom$.
The expansion is exponential in size, since the values up to $\denom'$ have to be enumerated explicitly.

We define $\alldivlong$ be the set of all constraints of the form
$\varx \eqdiv{\denom} \constc'$
where
$0 \leq \constc' < \denom$
and
$\varx \eqdiv{\denom} \constc$
appears directly in $\pres$ or in the expansion of the divisibility constraints of $\pres$.
This set will be used in the next sections.

\subsection{Soundness}
\label{sec:mondec-soundness}

We show that if
\[
    \forall \varx_1, \varx_2 \geq \bound .
        \forall \vecy .~~
            \ap{\samediv}{\varx_1, \varx_2, \vecy}
            \Rightarrow
            \brac{
                \ap{\pres}{\varx_1, \vecy}
                \iff
                \ap{\pres}{\varx_2, \vecy}
            }
\]
then $\pres$ is decomposable on $\varx$.
We do this by constructing the decomposition.

Although there are doubly exponentially many subsets
$\divset \subseteq \alldivlong$,
there are only exponentially many \emph{maximal consistent} subsets.
We implicitly restrict $\divset$ to such subsets.
This is because, for any $\denom$, there is no value of $\varx$ such that
$\varx \eqdiv{\denom} \constc$
and
$\varx \eqdiv{\denom} \constc'$
both hold with
$\constc \neq \constc'$ but $\constc, \constc' \in \{0, \ldots, k-1\}$.
For any maximal consistent set
$\divset \subseteq \alldivlong$,
let $\minsat{\divset}$ be the smallest integer greater than or equal to $\bound$ satisfying all constraints in $\divset$.
Note, since $\divset$ is maximal, a value that satisfies all constraints in $\divset$ also does not satisfy an constraints not in $\divset$.
The number~$\minsat{\divset}$ can be represented using polynomially many bits.

We can now decompose $\pres$ into
\[
    \brac{\begin{array}{c}
        \brac{\varx = 0 \land \ap{\pres}{0, \vecy}} \\
        \lor \cdots \lor \\
        \brac{\varx = \bound - 1 \land \ap{\pres}{\bound - 1, \vecy}}
    \end{array}}
    \lor
    \bigvee\limits_{
        \divset \subseteq \alldivlong
    }\brac{
        \varx \geq \bound
        \land
        \bigwedge\limits_{
            \varx \eqdiv{\denom} \constc
            \in
            \divset
        }
            \varx \eqdiv{\denom} \constc
        \land
        \ap{\pres}{\minsat{\divset}, \vecy}
    } \ .
\]
This formula is exponential in the size of $\pres$
if $D$ only ranges over the maximal consistent
subsets of $\alldivlong$.
For values of $\varx$ less than $\bound$, equivalence with the original formula is immediate.
For larger values, we use the fact that, from our original assumption, for any values $\varx_1$ and $\varx_2$ that satisfy the same divisibility constraints, we have
$\ap{\pres}{\varx_1, \vecy}$
iff
$\ap{\pres}{\varx_2, \vecy}$.
Hence, we can substitute the values $\minsat{\divset}$ in these cases.

\subsubsection{Example}

We return to
$\ap{\pres}{\varx, \vary, \varz} =
    \varx + 2 \vary \geq 5 \land \varz < 5 \land \varx \eqdiv{2} \vary$
and compute the decomposition on $\varx$.
Assuming $\bound$ is odd, the decomposition will be as follows.
In our presentation we slightly simplify the formula.
Strictly speaking
$\varx \eqdiv{2} \vary$
should be expanded to
$(\varx \eqdiv{2} 0 \land \vary \eqdiv{2} 0)
 \lor
 (\varx \eqdiv{2} 1 \land \vary \eqdiv{2} 1)$.
We simplify these to
$\vary \eqdiv{2} 0$
and
$\vary \eqdiv{2} 1$,
respectively, when instantiated with concrete values of $\varx$.
\[
    \begin{array}{c}
        \brac{
            x = 0
            \land
            \brac{0 + 2 \vary \geq 5 \land \varz < 5 \land \vary \eqdiv{2} 0}
        }
        \lor \\
        \brac{
            x = 1
            \land
            \brac{1 + 2 \vary \geq 5 \land \varz < 5 \land \vary \eqdiv{2} 1}
        } \\
        \lor \cdots \lor \\
        \brac{
            x = \bound - 1
            \land
            \brac{\bound - 1 + 2 \vary \geq 5 \land \varz < 5 \land \vary \eqdiv{2} 0}
        }
        \lor \\
        \brac{
            \brac{\varx \eqdiv{2} 0 \land \varx \geq \bound}
            \land
            \brac{\bound + 1 + 2 \vary \geq 5 \land \varz < 5 \land \vary \eqdiv{2} 0}
        }
        \lor \\
        \brac{
            \brac{\varx \eqdiv{2} 1 \land \varx \geq \bound}
            \land
            \brac{\bound + 2 \vary \geq 5 \land \varz < 5 \land \vary \eqdiv{2} 1}
        }
    \end{array}
\]

\subsection{Completeness}

We now show that every formula $\pres$ decomposable on $\varx$ satisfies
\[
    \forall \varx_1, \varx_2 \geq \bound .
        \forall \vecy .~~
            \ap{\samediv}{\varx_1, \varx_2, \vecy}
            \Rightarrow
            \brac{
                \ap{\pres}{\varx_1, \vecy}
                \iff
                \ap{\pres}{\varx_2, \vecy}
            } \ .
\]
We first show that some $\bound$ must exist.
Once the existence has been established, we can argue that it must be at most exponential in $\pres$.

\subsubsection{Existence of the Bound}

If
$\ap{\pres}{\varx, \vecy}$
is decomposable on $\varx$, then there is an equivalent formula
$\bigvee_{\idxi}
    \ap{\monpres_\idxi}{\varx}
    \land
    \ap{\varpres_\idxi}{\vecy}$.
It is known that every formula
$\ap{\monpres}{\varx}$
is satisfied by a finite union of arithmetic progressions
$\consta + \idxj \constb$.
Let $\bound$ be larger than the largest value of $\consta$ in the arithmetic progressions satisfying the
$\ap{\monpres_\idxi}{\varx}$.

We show when
$\ap{\samediv}{\varx_1, \varx_2, \vecy}$
then
$\ap{\pres}{\varx_1, \vecy}$
iff
$\ap{\pres}{\varx_2, \vecy}$
for all values $\varx_1, \varx_2 \geq \bound$ and $\vecy$.
Assume towards a contradiction that we have values
$\varx_1, \varx_2$
and a tuple of values $\vecy$ such that
$\ap{\samediv}{\varx_1, \varx_2, \vecy}$
and
$\ap{\pres}{\varx_1, \vecy}$, but not
$\ap{\pres}{\varx_2, \vecy}$.

Let $\denom$ be the product of all $\denom'$ appearing in some divisibility constraint
$x \eqdiv{\denom'} \constc$
in $\alldivlong$.
We know that there is some disjunct of the monadic decomposition such that
$\ap{\monpres}{\varx_1} \land \ap{\varpres}{\vecy}$
holds.
Moreover, let $\varx_1$ belong to the arithmetic progression
$\consta + \idxj \constb$.
Since
$\varx_1 \geq \bound > \consta$
we know that
$\ap{\monpres}{\varx'_1} \land \ap{\varpres}{\vecy}$
also holds for any
$\varx'_1 = \varx_1 + \idxj'\constb\denom$.
That is, we can pump $\varx_1$ by adding a multiple of $\constb\denom$, while staying in the same arithmetic progression and satisfying the same divisibility constraints.

Similarly, let $\constd$ be the product of all $\constb$ appearing in the (finite number of) arithmetic progressions that define the monadic decomposition of $\pres$, limited to disjuncts such that
$\ap{\varpres_\idxi}{\vecy}$
holds.
Since
$\ap{\pres}{\varx_2, \vecy}$
does not hold, then
$\ap{\pres}{\varx'_2, \vecy}$
also does not hold for any
$\varx'_2 = \varx_2 + \idxj\constd\denom$.
This means that we can pump $\varx_2$ staying outside of the arithmetic progressions defining permissible values of $\varx$ for the given values $\vecy$, whilst additionally satisfying the same divisibility constraints.

Now, for each value of $\varx'_1$ satisfying
$\ap{\pres}{\varx'_1, \vecy}$
we can consider the disjunctive normal form of $\pres$.
By expanding the divisibility constraints, a disjunct becomes a conjunction of terms of the form, where $\linearf$ represents some linear function on $\vecy$,
\begin{enumerate}
    \item
    $\consta \varx + \ap{\linearf}{\vecy} \leq \constc$
    or
    $\consta \varx + \ap{\linearf}{\vecy} \geq \constc$, or
    \item
    $\vary_\idxi \eqdiv{\denom'} \constc$
    or
    $\varx \eqdiv{\denom'} \constc$.
\end{enumerate}

Since there are infinitely many $\varx'_1$, we can choose one disjunct satisfied by infinitely many $\varx'_1$.
This means that for constraints of the form
$\consta \varx + \ap{\linearf}{\vecy} \leq \constc$
or
$\consta \varx + \ap{\linearf}{\vecy} \geq \constc$
with a non-zero $\consta$, then $\consta$ must be negative or positive respectively (or zero).
Otherwise, only a finite number of values of $\varx$ would be permitted.

We know that $\varx'_2$ and $\vecy$ do not satisfy the disjunct.
We argue that this is a contradiction by considering each term in turn.
Since there are infinitely many $\varx'_2$ we can assume without loss of generality that
$\varx'_2 > \varx'_1$.
\begin{enumerate}
    \item
    If
    $\consta \varx + \ap{\linearf}{\vecy} \leq \constc$
    (resp. $\consta \varx + \ap{\linearf}{\vecy} \geq \constc$)
    appears and is satisfied by $\varx'_1$, then $\consta$ must be negative or zero (resp. positive or zero) and $\varx'_2$ will also satisfy the atom.
    \item
    Atoms of the form $\vary_\idxi \eqdiv{\denom'} \constc$ do not distinguish values of $\varx$ and thus are satisfied for both $\varx'_1$ and $\varx'_2$.
    We cannot have
    $\varx'_1 \eqdiv{\denom'} \constc$
    but not
    $\varx'_2 \eqdiv{\denom'} \constc$
    since $\varx'_1$ and $\varx'_2$ satisfy the same divisibility constraints.
\end{enumerate}
Thus, it cannot be the case that $\varx'_1$ satisfies the disjunct, while $\varx'_2$ does not.
This is our required contradiction.
Hence, for all
$\varx_1, \varx_2 \geq \bound$
and $\vecy$ such that
$\ap{\samediv}{\varx_1, \varx_2, \vecy}$
it must be the case that
$\ap{\pres}{\varx_1, \vecy}$
iff
$\ap{\pres}{\varx_2, \vecy}$.
We have thus established the existence of a bound $\bound$.

\subsubsection{Size of the Bound}%
\label{sec:bound-size-mon}

We now argue that this bound is exponential in the size of $\pres$, and can thus be encoded in a polynomial number of bits.

Consider the formula that is essentially the negation of our property.
\[
    \ap{\badx}{\varx_1, \varx_2, \vecy}
    =
    \ap{\samediv}{\varx_1, \varx_2}
    \land
    \ap{\pres}{\varx_1, \vecy}
    \land
    \neg \ap{\pres}{\varx_2, \vecy} \ .
\]
There is some computable bound $\bound'$ exponential in the size of $\badx$ (and thus $\pres$) such that, if there exists
$\varx_1, \varx_2 \geq \bound'$
and some $\vecy$ such that
$\ap{\badx}{\varx_1, \varx_2, \vecy}$
holds, then there are infinitely many $\varx'_1$ and $\varx'_2$ such that for some $\vecy'$ we have that
$\ap{\badx}{\varx'_1, \varx'_2, \vecy'}$
holds.
An argument for the existence of this bound is given in
\shortlong{%
    the full version
}{%
    Appendix~\ref{sec:inf-sols}.
}
In short, we first convert the formula above into a disjunction of conjunctions of linear equalities, using a linear number of slack variables to encode inequalities and divisibility constraints.
Then, using a result of Chistikov and Haase~\cite{CH16}, we set
$\bound' = 2^{\constd\numof\varnumof + 3}$
where
    $\constd$ is the number of bits needed to encode the largest constant in the converted formula (polynomially related to the size of the formula above),
    $\numof$ is the maximum number of linear equalities in any disjunct, and
    $\varnumof$ is the number of variables (including slack variables).

Now, assume that the smallest $\bound$ is larger than $\bound'$.
That is
\[
    \forall \varx_1, \varx_2 \geq \bound .
        \forall \vecy .~~
            \ap{\samediv}{\varx_1, \varx_2, \vecy}
            \Rightarrow
            \brac{
                \ap{\pres}{\varx_1, \vecy}
                \iff
                \ap{\pres}{\varx_2, \vecy}
            }
\]
holds, but it does not hold that
\[
    \forall \varx_1, \varx_2 \geq \bound' .
        \forall \vecy .~~
            \ap{\samediv}{\varx_1, \varx_2, \vecy}
            \Rightarrow
            \brac{
                \ap{\pres}{\varx_1, \vecy}
                \iff
                \ap{\pres}{\varx_2, \vecy}
            }
\]
This implies there exists some
$\varx_1, \varx_2 \geq \bound'$
and $\vecy$ such that
$\ap{\badx}{\varx_1, \varx_2, \vecy}$
holds.
Thus, there are infinitely many such $\varx'_1$ and $\varx'_2$, contradicting the fact that all
$\varx'_1, \varx'_2 \geq \bound$
do not satisfy the property.
Thus, we take $\bound'$ as the value of $\bound$.
It is computable in polynomial time, exponential in size, and representable in a polynomial number of bits.


\section{Variadic Decomposability}
\label{sec:vardec}

We consider decomposition along several variables instead of just one.
In this section, we assume without loss of generality that $\pres$ is given in positive normal form and all (in)equalities rearranged into the form
$\sum_\idxi \consta_\idxi \varx_i \geq \constb$.
We may use negation $\neg \pres$ as a shorthand.
We require this form because later we use the set of all linear equations in the DNF of a formula.
Since negation alters the linear equations, it is more convenient to assume that negation has already been eliminated.

\subsection{$\Pi$-Decomposability}

As described in Section~\ref{sec:var-decomp-prelims}, we refine the notion of $\Pi$-decomposability to separate only a single set $Y_i$ in $\Pi = \{Y_1, \ldots, Y_n\}$.
Without loss of generality, we assume we are given a formula
$\ap{\pres}{\vecx, \vecy}$
and we separate the variables in $\vecx$ from $\vecy$.

In particular, given a formula
$\ap{\pres}{\vecx, \vecy}$
we aim to decompose the formula into
$
    \ap{\pres}{\vecx, \vecy}
    \equiv
    \bigvee\limits_{\idxj}
        \ap{\monpres_\idxj}{\vecx}
        \land
        \ap{\varpres_\idxj}{\vecy}
$
for some $\PresQF$ formulas $\monpres_\idxi$ and $\varpres_\idxi$.

\subsection{Decomposition}
\label{sec:variadic-decomposition}

We show that testing whether a given formula $\pres$ is variadic decomposable on $\vecx$ is in coNP.
This proves Theorem~\ref{thm:variadic} as the coNP lower bound follows from the monadic case.

\begin{lemma}[Decomposing on $\vecx$]
    Given a $\PresQF$ formula
    $\ap{\pres}{\vecx, \vecy}$
    there is a coNP algorithm to decide if $\pres$ is variadic decomposable on $\vecx$.
    Moreover, if a decomposition exists, it can be constructed in exponential-time and is exponential in size.
\end{lemma}

Let $\xfs$ be the set of all $\linearf$ such that
$\ap{\linearf}{\vecx} + \ap{\linearg}{\vecy} \geq \constb$
is a linear inequality appearing in $\pres$.
Our approach will divide the points of $\vecx$ into regions where all points within a region can be paired with the same values of $\vecy$ to satisfy the formula.
These regions are given by a bound $\bound$.
If
$\ap{\linearf}{\vecx}$
is within the bound, then two points $\vecx_1$ and $\vecx_2$ are in the same region if
$\ap{\linearf}{\vecx_1} = \ap{\linearf}{\vecx_2}$.
If two points are outside the bound, then by a pumping argument we can show that we have
$\ap{\pres}{\vecx_1, \vecy}$
iff
$\ap{\pres}{\vecx_2, \vecy}$.

Let
$\regmap = \tup{\regub, \regeq}$
be a partition of $\xfs$ into unbounded and bounded functions
(where $\regeq$ refers to equality being asserted over bounded functions as shown below).
Define for each
$\regmap = \tup{\regub, \regeq}$
\[
    \ap{\isreg{\regmap}}{\vecx_1, \vecx_2}
    \eqdef
    \brac{
        \bigwedge\limits_{\linearf \in \regeq}
            \ap{\linearf}{\vecx_1} = \ap{\linearf}{\vecx_2}
    }
    \ .
\]
Note, this formula intentionally does not say anything about the unbounded functions.
This is important when we need to derive a bound---such a derivation cannot use a pre-existing bound.

We also need to extend $\samediv$ to account for $\vecx_1$ and $\vecx_2$ being vectors.
This is a straightforward extension asserting that each variable in $\vecx_1$ satisfies the same divisibility constraints as its counterpart in $\vecx_2$.
Again, let $\alldivshort$ be the set of all divisibility constraints
$\consta \varz_1 \eqdiv{\denom} \constb \varz_2$
appearing (syntactically) in $\pres$.
Let $\varx_\idxi$, $\varx^1_\idxi$ and $\varx^2_\idxi$ denote the $\idxi$th variable of $\vecx$, $\vecx_1$, and $\vecx_2$ respectively.
Assume without loss of generality that variables in $\vecx$ always either appear on the left-hand side of a divisibility constraint (i.e.\ in the $\varz_1$ position) or on both sides.
Define
\begin{multline*}
    \ap{\samediv}{\varx_1, \varx_2, \vecy}
    = \\
    \hspace{-2ex}
    \bigwedge\limits_{\substack{
        \consta \varx_\idxi \eqdiv{\denom} \constb \varz
        \in
        \alldivshort, \\
        \varz \neq \varx_\idxj
    }}
    \hspace{-1ex}
    \brac{\begin{array}{c}
        \brac{\consta \varx^1_\idxi \eqdiv{\denom} \constb \varz} \\
        \iff \\
        \brac{\consta \varx^2_\idxi \eqdiv{\denom} \constb \varz}
    \end{array}}
    \land
    \hspace{-2ex}
    \bigwedge\limits_{\substack{
        \consta \varx_\idxi \eqdiv{\denom} \constb \varx_\idxj \\
        \in \\
        \alldivshort
    }}
    \hspace{-1ex}
    \brac{\begin{array}{c}
        \brac{\consta \varx^1_\idxi \eqdiv{\denom} \constb \varx^1_\idxj} \\
        \iff \\
        \brac{\consta \varx^2_\idxi \eqdiv{\denom} \constb \varx^2_\idxj}
    \end{array}}
    \land
    \hspace{-1ex}
    \bigwedge\limits_{\substack{
        \varx_\idxi \eqdiv{\denom} \constc \\
        \in \\
        \alldivshort
    }}\brac{\begin{array}{c}
        \brac{\varx^1_\idxi \eqdiv{\denom} \constc} \\
        \iff \\
        \brac{\varx^2_\idxi \eqdiv{\denom} \constc}
    \end{array}}
    \ .
\end{multline*}

Next, we introduce an operator for comparing a vector of variables with a bound.
For
$\consta \in \ints$
let
$\abs{\consta}$
denote the absolute value of $\consta$.
Given a bound $\bound$ and some
$\regmap = \tup{\regub, \regeq}$
let
\[
    \brac{\vecx \reggeq{\regmap} \bound}
    \eqdef
    \bigwedge\limits_{\linearf \in \regub}
        \abs{\ap{\linearf}{\vecx}} \geq \bound
    \land
    \bigwedge\limits_{\linearf \in \regeq}
        \abs{\ap{\linearf}{\vecx}} < \bound \ .
\]
We claim there is an exponential bound $\bound$ such that $\pres$ is variadic decomposable iff for all $\regmap$ we have
\begin{equation} \label{eq:decomptest} \tag{DC-$\regmap$}
    \forall \vecx_1, \vecx_2 \reggeq{\regmap} \bound\ .\ %
        \forall \vecy\ .\ %
            \brac{\begin{array}{c}
                \ap{\isreg{\regmap}}{\vecx_1, \vecx_2} \\
                \land \\
                \ap{\samediv}{\vecx_1, \vecx_2, \vecy}
            \end{array}}
            \Rightarrow
            \brac{\begin{array}{c}
                \ap{\pres}{\vecx_1, \vecy} \\
                \iff \\
                \ap{\pres}{\vecx_2, \vecy}
            \end{array}}
\end{equation}
Note, unsatisfiability can be tested in NP.
First guess $\regmap$, then guess $\vecx_1, \vecx_2, \vecy$.

We prove soundness of the claim in the next section.
Completeness is an extension of the argument for the monadic case and is given in
\shortlong{%
    the full version.
}{%
    Appendix~\ref{sec:variadic-completeness}.
}
In the monadic case, we were able to take some values of
$\varx_1, \varx_2 > \bound$
such that both satisfied the same divisibility constraints, but one value satisfied the formula while the other did not.
Since these values were large, we derived an infinite number of such value pairs with increasing values.
We then used these growing solutions to show that it was impossible for the value of $\varx_1$ to satisfy the formula, while the value of $\varx_2$ does not, as they were both beyond the distinguishing power of the linear inequalities.
The argument for the variadic case is similar, with the values of $\varx_1$ and $\varx_2$ being replaced by the values of
$\ap{\linearf}{\vecx_1}$
and
$\ap{\linearf}{\vecx_2}$.

\subsection{Soundness}

Assume there is an exponential bound $\bound$ such that for each $\regmap$, Equation~\ref{eq:decomptest} holds.
We show how to produce a decomposition.

As in the monadic case (Section~\ref{sec:mondec-soundness}), let $\alldivlong$ be the set of all constraints of the form
$\vecx_\idxi \eqdiv{\denom} \constc$
in the expansion of the divisibility constraints of $\pres$.
Observe again that there are only exponentially many maximal consistent subsets
$\divset \subseteq \alldivlong$.
For each $\divset$ fix a vector of values $\vecc_\divset$ that satisfies all constraints in $\divset$ and is encodable in a polynomial number of bits.
Furthermore, we define
\[
    \ap{\divsat{\divset}}{\vecz}
    \eqdef
    \bigwedge\limits_{
        \varx_\idxi \eqdiv{\denom} \constc
        \in \divset
    }
        \varz_\idxi \eqdiv{\denom} \constc
    \ .
\]

For each $\regmap$ and $\divset$ we can define an equivalence relation over values of $\vecx$ such that
$\vecx \reggeq{\regmap} \bound$
and
$\ap{\divsat{\divset}}{\vecx}$.
\[
    \brac{\vecx_1 \regequiv{\regmap}{\divset} \vecx_2}
    \eqdef
    \brac{
        \vecx_1 \reggeq{\regmap} \bound
        \land
        \vecx_2 \reggeq{\regmap} \bound
        \land
        \ap{\isreg{\regmap}}{\vecx_1, \vecx_2}
        \land
        \ap{\divsat{\divset}}{\vecx_1}
        \land
        \ap{\divsat{\divset}}{\vecx_2}
    } \ .
\]
Observe each equivalence relation has an exponential number of equivalence classes depending on the values of the bounded $\linearf$.
Let
$\eqreps{\regmap}{\divset}$
be a set of minimal representatives from each equivalence class such that each representative is representable in a polynomial number of bits.
These can be computed by solving an existential Presburger constraint for each set of values of the bounded $\linearf$.
In particular, for each
$\regmap = \tup{\regub, \regeq}$
and assignments
$\abs{\constc_\linearf} < \bound$
for each
$\linearf \in \regeq$, we select a solution to the equation
\[
    \vecx \reggeq{\regmap} \bound
    \land
    \bigwedge\limits_{\linearf \in \regeq}
        \ap{\linearf}{\vecx} = \constc_\linearf
    \land
    \ap{\divsat{\divset}}{\vecx}
\]
if such a solution exists.
If no such solution exists, the assignment can be ignored.

The decomposition is
\[
    \bigvee\limits_{\regmap}
        \bigvee\limits_{\divset}
            \bigvee\limits_{\vecc \in \eqreps{\regmap}{\divset}} \brac{
                \vecx \reggeq{\regmap} \bound
                \land
                \ap{\isreg{\regmap}}{\vecx, \vecc}
                \land
                \ap{\divsat{\divset}}{\vecx}
                \land
                \ap{\pres}{\vecc, \vecy}
            } \ .
\]
The correctness of this decomposition follows from the Equations~\ref{eq:decomptest}.
For any values $\vecc_{\vecx}$ and $\vecc_{\vecy}$ of $\vecx$ and $\vecy$, first assume
$\ap{\pres}{\vecc_{\vecx}, \vecc_{\vecy}}$
holds.
Since there is some disjunct in the decomposition for which it holds that
$\vecc_{\vecx} \reggeq{\regmap} \bound
 \land
 \ap{\isreg{\regmap}}{\vecc_{\vecx}, \vecc}
 \land
 \ap{\divsat{\divset}}{\vecx}$
then, by applying Equation~\ref{eq:decomptest} we get
$\ap{\pres}{\vecc, \vecc_{\vecy}}$
as required.
Conversely, if some disjunct of the decomposition holds, we can apply Equation~\ref{eq:decomptest} and obtain
$\ap{\pres}{\vecc_{\vecx}, \vecc_{\vecy}}$.

\section{Applications of Decomposition}
\label{sec:app}

\subsection{Monadic Decomposition in String Solving}

The development of effective techniques for solving string constraints
has received a lot of attention over the last years, motivated by
applications ranging from program
verification~\cite{Abdulla14,philipp-survey} and security
analysis~\cite{Berkeley-JavaScript,S3} to the analysis of access policies
of cloud services~\cite{DBLP:conf/fmcad/BackesBCDGLRTV18}. Strings
give rise to a rich theory that may combine, depending on the
studied fragment,
(i)~word equations, i.e., equations over the free monoid generated by
some finite (but often large) alphabet,
(ii)~regular expression constraints,
(iii)~transduction, i.e., constraints described by finite-state
automata with multiple tracks,
(iv)~conversion functions, e.g.\ between integer variables and
strings encoding numbers in binary or decimal notation,
(v)~length constraints, i.e., arithmetic constraints on the length of
strings.

The handling of length constraints has turned out to be particularly
challenging in this context, both practically and theoretically.  Even
for the combination of word equations (or even just quadratic word
equations) with length constraints, decidability of the
(quantifier-free) theory is a long-standing open
problem~\cite{DBLP:conf/atva/LinM18}. At the same time, length
constraints are quite frequently used in applications; they are
needed, for instance, when encoding operations like \textsf{indexof}
or \textsf{substring}, or also when splitting a string into the parts
separated by some delimiter. In standard benchmark libraries
for string constraints, like the Kaluza
set~\cite{Berkeley-JavaScript}, benchmarks with length constraints
occur in large numbers.

The notion of monadic decomposition is in this setting important, since
any \emph{monadic} length constraint (in Presburger arithmetic) can be
reduced to a Boolean combination of regular expression constraints,
and is therefore easier to handle than the general case.

\begin{proposition}
  Satisfiability of a quantifier-free
  formula~$\phi = \phi_{\text{eq}} \wedge \phi_{\text{regex}} \wedge
  \phi_{\text{len}}$ consisting of word equations, regular expression
  constraints, and \emph{monadically decomposable} length constraints is
  decidable.
\end{proposition}
\begin{proof}
  Suppose $w_1, \ldots, w_n$ are the string variables occurring in
  $\phi$, and $|w_1|, \ldots, |w_n|$ the terms representing their
  length.  A decision procedure can first compute a monadic
  representation~$\phi'_{\text{len}}$ of $\phi_{\text{len}}$ over
  lengths~$|w_1|, \ldots, |w_n|$, and then turn each
  atom~$\Delta(|w_i|)$ in $\phi'_{\text{len}}$ into an equivalent
  regular membership constraint~$w_i \in {\cal L}_\Delta$. This is
  possible because the Presburger formula~$\Delta$ can be represented
  as a semi-linear set, which can directly be translated to a regular
  expression. Decidability follows from the decidability of word
  equations combined with regular expression
  constraints~\cite{diekert}.
  \qed
\end{proof}

\begin{table}[tb]
  \caption{Statistics about the Kaluza
    benchmarks~\cite{Berkeley-JavaScript}.  It should be noted (and is
    well-known~\cite{cvc4}) that the categories ``sat'' and ``unsat''
    do not (always) imply the status of the benchmarks, they only
    represent the way the benchmarks were organised by the Kaluza
    authors.}
  \label{tab:kaluza}

  \begin{center}
    \begin{tabular}{l*{4}{@{\qquad}c}}
      \multirow{2}{*}{\textbf{Folder}} & \multirow{2}{*}{\#Benchmarks}
      & Benchmarks & Decomposition & Decomposition
      \\
      &  & with \texttt{str.len} & checks & checks succeeded
      \\\hline
      sat/small & 19804 & 2185 & 2183 & 2155
      \\\hline
      sat/big & 1741 & 1318 & 1317 & 56
      \\\hline
      unsat/small & 11365 & 3910 & 2919 & 2919
      \\\hline
      unsat/big & 14374 & 13813 & 6786 & 3362
      \\\hline
      \textbf{Total} & 47284 & 21226 & 13205 & 8492
    \end{tabular}
  \end{center}
\end{table}

This motivates the use of monadic decomposition as a standard
pre-processing step in string solvers, transforming away those length
constraints that can be turned into monadic form. To evaluate the
effectiveness of such an optimisation, we implemented the
decomposition check defined in Section~\ref{sec:monadicUpper}, and
used it within the string SMT solver OSTRICH~\cite{CHLRW18} to
determine the number of Kaluza benchmarks with monadic
decomposable length constraints.\footnote{Branch ``modec'' of
 \url{https://github.com/uuverifiers/ostrich}, which also contains
detailed logs of the experiments.}
The results are summarised in
Table~\ref{tab:kaluza}:
\begin{itemize}
\item Of altogether 47\,284 benchmarks, 21\,226 contain the
  \texttt{str.len} function, and therefore length constraints. This
  number was determined by a simple textual analysis of the
  benchmarks.
\item Running our decomposition check in OSTRICH, in 13\,205 of the
  21\,226 cases length constraints were found that could be
  analysed. The remaining 8\,021 problems were proven unsatisfiable
  without ever reaching the string theory solver in OSTRICH, i.e., as
  a result of pre-processing the input formula, or because Boolean
  reasoning discovered obvious inconsistencies in the problems.
\item In 8\,492 of the 13\,205 cases, all analysed length constraints
  were found to be monadically decomposable; 4\,713 of the benchmarks
  contained length constraints that could not be decomposed.
\end{itemize}

This means that 42\,571 of the Kaluza benchmarks (slightly more than
90\%) do in principle not require support for length constraints in a
string solver, either because there are no length constraints, or
because length constraints can be decomposed and then turned into
regular expression constraints.

Even with a largely unoptimised implementation, the time required to
check whether length constraints can be decomposed was negligible in
case of the Kaluza benchmarks, with the longest check requiring
2.1~seconds (on an AMD Opteron 2220 SE machine). The maximum number of
variables in a length constraint was 140.

\subsection{Variadic Decomposition in Quantifier Elimination}

A second natural application of decomposition is \emph{quantifier
  elimination,} i.e., the problem of deriving an equivalent
quantifier-free formula~$\phi'$ for a given formula~$\phi$ with
quantifiers.  In Presburger arithmetic, for a
formula~$\phi = \exists x_1, \ldots, x_n.\, \psi$ with $n$
quantifiers but no quantifier alternations, quantifier elimination in
the worst case causes a doubly-exponential increase in formula
size~\cite{Weis97}.

Variadic decomposition can be used to eliminate quantifiers with a
smaller worst-case increase in size, provided that the matrix of a
quantifier formula can be decomposed. Suppose
$\phi = \exists \bar x.\, \psi(\bar x, \bar y)$ is given and
$\psi$ is variadic decomposable on $\bar x$, i.e.,
\begin{equation*}
  \ap{\psi}{\bar x, \bar y}
  ~\equiv~
  \bigvee\limits_{\idxj}
  \ap{\monpres_\idxj}{\bar x}
  \land
  \ap{\varpres_\idxj}{\bar y}
\end{equation*}
This means that the existential quantifiers can be distributed over
the disjunction, and their elimination turns into a simpler
satisfiability check:
\begin{equation*}
  \exists \bar x.\, \ap{\psi}{\bar x, \bar y}
  ~\equiv~
  \bigvee\limits_{\idxj}
  \exists \bar x.\,
  \ap{\monpres_\idxj}{\bar x}
  \land
  \ap{\varpres_\idxj}{\bar y}
  ~\equiv~
  \bigvee\limits_{\idxj:~  \ap{\monpres_\idxj}{\bar x} \text{~is sat}}
  \ap{\varpres_\idxj}{\bar y}
\end{equation*}
Universal quantifiers can be handled in a similar way by negating the
matrix first.

\begin{proposition}
  Take a
  formula~$\phi(\bar y) = \exists \bar x.\, \psi(\bar x, \bar y)$ in
  Presburger arithmetic in which $\psi$ is quantifier-free and
  variadic decomposable on $\bar x$. Then there is a quantifier-free
  formula~$\phi'(\bar y)$ that is equivalent to $\phi$ and at
  most singly-exponentially bigger than $\phi$.
\end{proposition}

Checking whether a formula can be decomposed is therefore a simple
optimisation that can be added to any quantifier elimination
procedure for Presburger arithmetic.


\section{Conclusion and Future Work}
\label{sec:conclusion}

We have shown that the monadic and variadic decomposability problem for $\PresQF$ is coNP-complete.
Moreover, when a decomposition exists, it is at most exponential in size and can be computed in exponential time.
This formula size is tight for decompositions presented in either disjunctive or conjunctive normal form.

We gave two applications of our results.
The first was in string constraint solving.
In program analysis, string constraints are often mixed with numerical constraints on the lengths of the strings (for example, via the \textsf{indexOf} function).
Length constraints significantly complicate the analysis of strings.
However, if the string constraints permit a monadic decomposition, they may be reduced to regular constraints and thus eliminated.
We analysed the well-known Kaluza benchmarks and showed that less than 10\% of the benchmarks contained length constraints that could not be decomposed.

For the second application, we showed that the doubly exponential blow-up caused by quantifier elimination can be limited to a singly exponential blow up whenever the formula is decomposable on the quantified variables.
Thus, variadic decomposition can form an optimisation step in a quantifier elimination algorithm.

Interesting problems are opened up by our results. It would be interesting to
study lower bounds for 
general boolean formulas. If smaller decompositions are possible, they would be 
useful for applications in string solving.

\OMIT{
First, while we show that decompositions in DNF or CNF are necessarily exponential, it is not clear whether the lower bound also holds for decompositions given as arbitrary formulas.
This is because divisibility constraints and the Chinese Remainder Theorem may be used to avoid the exponential blow-up in our counter example.
}

Second, we may consider variadic decomposition where a partition $\Pi$ is not
given as part of the input. 
Instead, one must check whether a $\Pi$-decomposition exists for some non-trivial $\Pi$.
This variant of the problem has a simple $\Sigma^P_2$ algorithm that first guesses some $\Pi$ and then verifies $\Pi$-decomposability.
However, the only known lower bound is coNP, which follows the same argument as
monadic decomposability. A better algorithm would not improve the
worst-case complexity for our quantifier elimination application, but it might 
provide a way to quickly identify a subset of a block of quantifiers that can
be eliminated quickly with $\Pi$-decompositions.


\paragraph*{Acknowledgments}

We thank Christoph Haase, Leonid Libkin, and Pascal Bergstr\"a{\ss}er for their 
help during the preparation of this work. 
Matthew Hague is supported by EPSRC [EP/T00021X/1]. 
Anthony Lin is supported by the European Research Council (ERC) under the 
European
Union's Horizon 2020 research and innovation programme (grant agreement no
    759969), and
by Max-Planck Fellowship.
Philipp R\"ummer is supported by the
Swedish Research Council (VR) under grant 2018-04727, and by the
Swedish Foundation for Strategic Research (SSF) under the project
WebSec (Ref.\ RIT17-0011).
Zhilin Wu is partially supported by the NSFC grant No. 61872340, Guangdong Science and Technology Department grant (No. 2018B010107004),  and the INRIA-CAS joint research project VIP.

\bibliographystyle{plain}
\bibliography{references}

\appendix

\section{Infinite Solutions of Presburger Formulas}
\label{sec:inf-sols}

For a given $\PresQF$ formula $\pres$, we show that there is a bound $\bound$ exponential in the size of $\pres$ such that if
$\ap{\pres}{\varx_1, \ldots, \varx_\numof}$
holds for some value of $\varx_\idxi$ greater than $\bound$ (for some $\idxi$), then there are infinitely many satisfying assignments.
This fact is quite standard, but we explicate it here for our particular definition of Presburger formulas.

Given $\pres$ we replace all terms
$\consta \varx \eqdiv{\denom} \constb \vary$
with
$\consta \varx = \varz + \denom \varx'
 \land
 \constb \vary = \varz + \denom \vary'$
for fresh variables $\varz$, $\varx$, and $\vary$, and
all terms
$\varx \eqdiv{\denom} \constc$
with
$\varx = \constc + \denom\varx'$
for some fresh variable $\varx'$.
This leaves us with only equality and inequality constraints in the formula.
We can replace inequalities with equalities via the introduction of a linear number of slack variables.

Next, observe that if we convert the formula to disjunctive normal form, we have a finite union of conjunctions of linear equalities.
Now, we rephrase a Proposition from Chistikov and Haase~\cite{CH16} -- which follows from Pottier~\cite{P91} -- that gives bounds on the solutions to linear equalities.

We first describe some notation.
Given finite sets of vectors
$\baseset, \periodset \subset \nats^\numof$
let
\[
    \semilin{\baseset}{\periodset} =
    \setcomp{
        \veca +
        \idxi_1 \vecp_1 +
        \cdots +
        \idxi_\varnumof \vecp_\varnumof
    }{
        \veca \in \baseset,
        \vecp_1, \ldots, \vecp_\varnumof \in \periodset,
        \idxi_1, \ldots, \idxi_\varnumof \in \nats
    } \ .
\]
For a vector
$\veca = \tup{\consta_1, \ldots, \consta_\numof}$
We write $\maxiof{\veca}$ to denote the largest $\consta_\idxi$.
For a finite set
$\baseset \subseteq \nats^\numof$
we write $\maxiof{\baseset}$ to denote the largest value of $\maxiof{\veca}$ for all
$\veca \in \baseset$.
Finally, given a conjunction of linear equalities (or a Presburger formula) $\varpres$ with $\numof$ variables, we write
\[
    \solsof{\varpres} =
    \setcomp{
        \tup{\consta_1, \ldots, \consta_\numof}
        \in \nats^\numof
    }{
        \ap{\varpres}{\consta_1, \ldots, \consta_\numof}
        \text{ holds}
    } \ .
\]

\begin{proposition}[\cite{CH16}]
    Given a conjunction of $\numof$ linear equalities $\varpres$ over $\varnumof$ variables such that $\consta$ is the largest constant in $\varpres$ then
    $\solsof{\varpres} = \semilin{\baseset}{\periodset}$
    for some
    $\baseset, \periodset \in \nats^\varnumof$
    where
    \begin{enumerate}
    \item
        $\maxiof{\baseset} \leq
         ((\numof + 2)\consta + 1)^\varnumof$, and
    \item
        $\maxiof{\periodset} \leq
         (\numof \consta + 1)^\varnumof$.
    \end{enumerate}
\end{proposition}

Since constants are encoded in binary, the largest constant in the formula derived from $\pres$ is exponential in the size $\constd$ of $\pres$.
After the expansion of divisibility constraints, the number of bits needed to encode the largest constant will be bound by
$\ap{\polyf}{\constd}$
for some polynomial $\polyf$.
After introducing fresh variables to remove divisibility constraints, and adding slack variables to remove inequalities, the number of variables $\varnumof$ is polynomially related to the size of $\pres$.
Similarly, the number of clauses $\numof$ in any disjunct in the disjunctive normal form is also polynomially related to the size of $\pres$.
Thus,
$\maxiof{\baseset} \leq
 ((\numof + 2)2^{\ap{\polyf}{\constd}} + 1)^\varnumof \leq
 2^{\ap{\polyf}{\constd}\numof\varnumof + 3} =
 \bound$.
This is exponential in the size of $\pres$.

Now, assume we have some
$\veca =
 \tup{\consta_1, \ldots, \consta_\varnumof}
 \in
 \solsof{\pres}$
such that for some $\idxi$ we have
$\consta_\idxi > \bound$.
Let $\veca$ satisfy disjunct $\varpres$ of the transformation of $\pres$.
We have
$\veca \in \semilin{\baseset}{\periodset}$
for some $\baseset$ and $\periodset$.
That
$\consta_\idxi > \bound$
implies $\periodset$ is non-empty as $\veca$ cannot be contained in $\baseset$.
Moreover, there must be some
$\idxi_1 \vecp_1 +
 \cdots +
 \idxi_\varnumof \vecp_\varnumof$
that is non-zero in all components $\idxi$ such that $\consta_\idxi > \bound$ with
$\veca = \vecb + \idxi_1 \vecp_1 + \cdots + \idxi_\varnumof \vecp_\varnumof$
for some
$\vecb \in \baseset$.
By the definition of
$\semilin{\baseset}{\periodset}$
we know that
$\vecb + \idxj \brac{
     \idxi_1 \vecp_1 +
     \cdots +
     \idxi_\varnumof \vecp_\varnumof
 }
 \in
 \solsof{\pres}$
for all $\idxj$.
Thus, there are infinitely many solutions.

\section{Completeness of Variadic Decomposition}
\label{sec:variadic-completeness}

We prove completeness of the claim that there is an exponential bound $\bound$ such that $\ap{\pres}{\vecx, \vecy}$ is variadic decomposable on $\vecx$ iff for all $\regmap$ we have
\[
    \forall \vecx_1, \vecx_2 \reggeq{\regmap} \bound\ .\ %
        \forall \vecy\ .\ %
            \brac{\begin{array}{c}
                \ap{\isreg{\regmap}}{\vecx_1, \vecx_2} \\
                \land \\
                \ap{\samediv}{\vecx_1, \vecx_2, \vecy}
            \end{array}}
            \Rightarrow
            \brac{\begin{array}{c}
                \ap{\pres}{\vecx_1, \vecy} \\
                \iff \\
                \ap{\pres}{\vecx_2, \vecy}
            \end{array}}
    \ .
\]
We first prove that such a $\bound$ exists.
Then we prove it is exponential.

\subsubsection{Existence of a Bound}

Assume that $\pres$ is variadic decomposable on $\vecx$.
We show that Equation~\ref{eq:decomptest} holds for each $\regmap$.
During this section we will also show that $\bound$ exists.

We introduce a number of auxiliary variables
$\auxfplus{\linearf}$
and
$\auxfminus{\linearf}$
to track the value of each
$\linearf \in \xfs$.
Using these, we aim to prove that if the size of
$\ap{\linearf}{\vecx}$
is larger than $\bound$, then we can produce arbitrarily large values of
$\ap{\linearf}{\vecx}$.
Let
$\vecaux = \brac{\auxfplus{\linearf}, \auxfminus{\linearf}}_{\linearf \in \xfs}$.
We define
\[
    \ap{\auxeq}{\vecx, \vecaux} =
        \bigwedge\limits_{\linearf \in \xfs}
            \brac{
                \ap{\linearf}{\vecx} = \auxfplus{\linearf}
                \land
                \auxfminus{\linearf} = 0
            }
            \lor
            \brac{
                \ap{\linearf}{\vecx} = -\auxfminus{\linearf}
                \land
                \auxfplus{\linearf} = 0
            }\ .
\]

Towards a contradiction, assume there is some $\regmap$ such that Equation~\ref{eq:decomptest} does not hold.
That is, there are
$\vecc_1, \vecc_2 \reggeq{\regmap} \bound$
and some $\vecd$ such that
$\ap{\isreg{\regmap}}{\vecc_1, \vecc_2}$
and
$\ap{\samediv}{\vecc_1, \vecc_2, \vecd}$
and
$\ap{\pres}{\vecc_1, \vecd}
 \land
 \neg \ap{\pres}{\vecc_2, \vecd}$.
This implies the existence of $\veceaux^1$ and $\veceaux^2$ such that
\[
    \ap{\isreg{\regmap}}{\vecc_1, \vecc_2}
    \land
    \ap{\samediv}{\vecc_1, \vecc_2, \vecd}
    \land
    \ap{\auxeq}{\vecc_1, \veceaux^1}
    \land
    \ap{\auxeq}{\vecc_2, \veceaux^2}
    \land
    \ap{\pres}{\vecc_1, \vecd}
    \land
    \neg \ap{\pres}{\vecc_2, \vecd} \ .
\]
Since $\pres$ is variadic decomposable, it is equivalent to some formula
\[
    \bigvee\limits_{\idxj}
        \ap{\monpres_\idxj}{\vecx}
        \land
        \ap{\varpres_\idxj}{\vecy} \ .
\]
Observe also that the negation of $\pres$ is thus also equivalent to some decomposed formula
\[
    \bigvee\limits_{\idxj}
        \ap{\monpres'_\idxj}{\vecx}
        \land
        \ap{\varpres'_\idxj}{\vecy} \ .
\]
Furthermore, for any assignment $\vecd'$ to $\vecy$ there is a polynomially encodable $\vecd''$ (in the size of $\pres$) such that
$\ap{\samediv}{\vecx_1, \vecx_2, \vecd'}$
iff
$\ap{\samediv}{\vecx_1, \vecx_2, \vecd''}$.
This is because $\samediv$ encodes divisibility constraints only.

Hence, we can replace $\pres$ with its decomposition and $\vecd$ with a polynomially encodable $\vecd'$ and conclude that there is some formula
\[
    \ap{\isreg{\regmap}}{\vecx_1, \vecx_2}
    \land
    \ap{\samediv}{\vecx_1, \vecx_2, \vecd'}
    \land
    \ap{\auxeq}{\vecx_1, \vecaux^1}
    \land
    \ap{\auxeq}{\vecx_2, \vecaux^2}
    \land
    \ap{\monpres_\idxj}{\vecx_1}
    \land
    \ap{\monpres'_{\idxj'}}{\vecx_2}
\]
that also holds.
Due to the way we constructed the formula, there are only finitely many such formulas that we may consider.

Let $\bound$ be a bound be such that for all such formulas, the existence of a solution
$\tup{\vecc_1, \vecc_2, \veceaux^1, \veceaux^2}$
with some components larger than $\bound$ implies the existence of an infinite number of solutions.
In particular, we may assume that the solutions are growing in all components above $\bound$.
The existence of such a bound is argued in
\shortlong{%
    the full version
}{%
    Appendix~\ref{sec:inf-sols}
}
for a particular formula.
Here we are applying this argument to all formulas of the form above that can be constructed from the decomposition of $\pres$ (and its negation).
Note, these solutions may also be growing on some components below $\bound$.

In the simplest case, suppose $\vecc_1$ does not exceed $\bound$ on any component.
In this case $\regub$ is empty and thus
$\ap{\linearf}{\vecc_1} = \ap{\linearf}{\vecc_2}$
for all
$\linearf \in \xfs$.
Together with
$\ap{\samediv}{\vecc_1, \vecc_2, \vecd'}$
we obtain a contradiction against
$\ap{\pres}{\vecc_1, \vecd}
 \land
 \neg \ap{\pres}{\vecc_2, \vecd}$.

Now suppose some component of $\vecc_1$ exceeds the bound.
This implies the same component of $\vecc_2$ also exceeds the bound
(this is implied by $\reggeq{\regmap}$ as all components below the bound have
$\linearf \in \regeq$).
In this case, there are infinitely many
$\tup{\vecc'_1, \vecc'_2, \veceaux', \veceaux''}$
such that
\[
    \ap{\isreg{\regmap}}{\vecc'_1, \vecc'_2}
    \land
    \ap{\samediv}{\vecc'_1, \vecc'_2, \vecd'}
    \land
    \ap{\auxeq}{\vecc'_1, \veceaux'}
    \land
    \ap{\auxeq}{\vecc'_2, \veceaux''}
    \land
    \ap{\monpres_\idxj}{\vecc'_1}
    \land
    \ap{\monpres'_{\idxj'}}{\vecc'_2}
\]
holds.
That is, there are infinitely many $\vecc'_1$ that satisfy $\pres$ and infinitely many $\vecc'_2$ that do not, for a given $\vecd$.
Moreover, these solutions are growing in all components above $\bound$ (and possibly others).

Consider the DNF of $\pres$.
There is some disjunct that is satisfied by an infinite number of the $\vecc'_1$ above (together with $\vecd$).
Similarly, all $\vecc'_2$ do not satisfy the disjunct.
Thus, we can assume we can pick two elements
$(\vecc^l_1, \vecc^l_2)$
and
$(\vecc''_1, \vecc''_2)$
from this sequence such that for all growing components $\linearf$ we have
$\abs{\ap{\linearf}{\vecc^l_1}} < \abs{\ap{\linearf}{\vecc''_2}}$.
Fix such a disjunct and $\vecc^l_1$, $\vecc''_1$, and $\vecc''_2$ and also $\vecd$.
Note, we do not need $\vecc^l_2$ in the following proof.
The important property is that $\vecc^l_1, \vecd$ is an assignment satisfying the disjunct where $\vecc^l_1$ is smaller than $\vecc''_2$.

By definition, the satisfied disjunct is a conjunction of divisibility constraints and linear inequalities
\[
    \ap{\linearf}{\vecx} + \ap{\linearg}{\vecy} \geq \constb \ .
\]

First, consider the possibility that there is some $\linearf$ such that the value of the
$\ap{\linearf}{\vecx}$
is growing \emph{negatively}.
This case cannot occur since
$\ap{\linearg}{\vecd}$
is fixed, and hence such a growing sequence must eventually fail to satisfy the disjunct, contradicting our assumptions.

Now, there are two cases.
Note, these depend on the original assignment $\vecc_1$ and not the grown assignments.
\begin{itemize}
\item
    If
    $\abs{\ap{\linearf}{\vecc_1}} < \bound$
    then from
    $\ap{\isreg{\regmap}}{\vecc''_1, \vecc''_2}$
    we have
    $\ap{\linearf}{\vecc''_1} = \ap{\linearf}{\vecc''_2}$
    and since
    $\ap{\linearf}{\vecc''_1} + \ap{\linearg}{\vecd} \geq \constb$.
    we have
    $\ap{\linearf}{\vecc''_2} + \ap{\linearg}{\vecd} \geq \constb$.
    This remains true even if the $\linearf$ component is growing.

\item
    If
    $\abs{\ap{\linearf}{\vecc_1}} \geq \bound$
    then the $\linearf$ component must be growing, and hence positive.
    In this case, since
    $\ap{\linearf}{\vecc^l_1} + \ap{\linearg}{\vecd} \geq \constb$
    we must have
    $\ap{\linearf}{\vecc''_2} + \ap{\linearg}{\vecd} \geq \constb$.
\end{itemize}

Thus, $\vecc''_2$ and $\vecd$ satisfy all inequalities in the disjunct.
From $\samediv$ we know they satisfy the same divisibility constraints.
Consequently $\vecc''_2$ and $\vecd$ satisfy $\pres$, which is a contradiction.
This proves that all Equations~\ref{eq:decomptest} must be satisfied.

\subsubsection{The Bound is Exponential}
\label{sec:variadic-exp-bound}

Again, we introduce auxiliary variables
$\auxfplus{\linearf}$
and
$\auxfminus{\linearf}$
to track the value of each
$\linearf \in \xfs$.
For each
$\regmap = \tup{\regub, \regeq}$
consider the equation
\[
    \ap{\isreg{\regmap}}{\vecx_1, \vecx_2}
    \land
    \ap{\samediv}{\vecx_1, \vecx_2, \vecy}
    \land
    \ap{\pres}{\vecx_1, \vecy}
    \land
    \neg \ap{\pres}{\vecx_2, \vecy}
    \land
    \ap{\auxeq}{\vecx_1, \vecaux^1}
    \land
    \ap{\auxeq}{\vecx_2, \vecaux^2} \ .
\]
We know (see
\shortlong{%
    the full version)
}{%
    Appendix~\ref{sec:inf-sols})
}
that there is an exponential bound
$\bound_\regmap$
such that if $\regub$ is empty, there is a solution with all components less than $\bound_\regmap$.
Otherwise, if $\regub$ is non-empty and there is a solution
$\tup{\vecc_1, \vecc_2, \vecd}$
with
$\vecc_1, \vecc_2 \reggeq{\regmap} \bound_\regmap$
then there are infinitely many solutions, growing in all components above $\bound_\regmap$.
Let
$\regmap' = \tup{\regub', \regeq'}$
be the partition between growing components and stable components in these infinite solutions.
Since these solutions are growing in $\regub'$, there must be a solution to
\[
    \vecx_1, \vecx_2 \reggeq{\regmap'} \bound
    \land
    \ap{\isreg{\regmap'}}{\vecx_1, \vecx_2}
    \land
    \ap{\samediv}{\vecx_1, \vecx_2, \vecy}
    \land
    \ap{\pres}{\vecx_1, \vecy}
    \land
    \neg \ap{\pres}{\vecx_2, \vecy} \ .
\]

Let $\bound'$ be the largest $\bound_\regmap$.
From the above, it follows that a solution with bound $\bound$ implies a solution with $\bound'$ and vice-versa.
Hence, the bound $\bound$ can be limited to be at most exponential.

\end{document}